\documentclass[11pt]{article}
\usepackage{fullpage}
\usepackage[utf8]{inputenc}
\usepackage[linesnumbered,ruled,vlined]{algorithm2e}
\usepackage{amsfonts}
\usepackage{graphicx}
\usepackage{amsthm}
\usepackage{mathrsfs}
\usepackage{amssymb}
\usepackage{hyperref}
\usepackage{amsmath}
\usepackage{enumitem}
\usepackage{xcolor}
\usepackage{xspace}
\usepackage{cleveref}
\usepackage[numbers]{natbib}

\usepackage{authblk}
\usepackage{booktabs}
\usepackage{pbox}
\usepackage{multirow}
\usepackage{makecell}

\newtheorem{theorem}{Theorem}
\newtheorem{lemma}{Lemma}

\newtheorem{definition}{Definition}

\newcommand{\set}[1]{\left\{ #1 \right\}}

\newcommand{\norm}[1]{\left|\left|#1\right|\right|}

\newcommand{\floor}[1]{\left\lfloor #1 \right\rfloor}
\newcommand{\ip}[1]{\langle #1 \rangle}

\newcommand{\real}{\mathbb{R}}

\renewcommand{\natural}{\mathbb{N}}

\let\phi\varphi

\DeclareMathOperator{\poly}{poly}

\def\e#1{\emph{#1}}

\def\T{\mathbf{T}}
\def\P{\mathbf{P}}

\def\R{\mathbf{R}}
\def\A{\mathbf{A}}
\def\J{\mathbf{J}}

\def\O{\mathcal{O}}

\def\cvec#1{\mathbf{#1}}

\def\nw{\mathsf{NW}}
\def\pw{\mathsf{PW}}
\def\angs#1{\mathord{\langle{#1\rangle}}}
\def\nwpar#1{\nw\angs{#1}}
\def\pwpar#1{\pw\angs{#1}}

\def\Cleft{C_{\mathsf{lft}}}
\def\Cright{C_{\mathsf{rgt}}}
\def\Cmiddle{C_{\mathsf{mid}}}
\def\cleft{c_{\mathsf{lft}}}
\def\cright{c_{\mathsf{rgt}}}
\def\cmiddle{c_{\mathsf{mid}}}
\def\Ctop{C_{\mathsf{top}}}
\def\Cbottom{C_{\mathsf{bot}}}

\def\dtilde{\Tilde{d}}

\newcommand{\eat}[1]{}

\author[1]{\large Aviram Imber}
\author[2]{\large Jonas Israel}
\author[2,3]{\large Markus Brill}
\author[1]{\large Hadas Shachnai}
\author[1]{\large Benny Kimelfeld}

\affil[1]{\normalsize Technion -- Israel Institute of Technology, Haifa, Israel}
\affil[2]{\normalsize Research Group Efficient Algorithms, TU Berlin, Germany}
\affil[3]{\normalsize Department of Computer Science, University of Warwick, UK}

\date{}

\begin{document}

\title{Spatial Voting with Incomplete Voter Information}
\maketitle

\begin{abstract}
       We consider spatial voting where candidates are located in the Euclidean $d$-dimensional space, and each voter ranks candidates based on their distance from the voter's ideal point. We explore the case where information about the location of voters' ideal points is incomplete: for each dimension, we are given an interval of possible values. We study the computational complexity of finding the possible and necessary winners for positional scoring rules. Our results show that we retain tractable cases of the classic model where voters have partial-order preferences. Moreover, we show that there are positional scoring rules under which the possible-winner problem is intractable for partial orders, but tractable in the one-dimensional spatial setting. 
       We also consider approval voting in this setting. We show that for up to two dimensions, the necessary-winner problem is tractable, while the possible-winner problem is hard for any number of dimensions.
\end{abstract}

\section{Introduction}
In the spatial model of voting \cite{enelow1984spatial,Mill95a}, both candidates and voters are associated with points in the $d$-dimensional Euclidean space $\mathbb{R}^d$. It is assumed that the locations of candidates and voters correspond to their ``ideal points'' and that each voter's preferences over the candidates can be inferred from the Euclidean distance between the candidates' and the voter's ideal points. For example, the location of a candidate or voter in $\mathbb{R}^d$ could reflect the stance (or opinion) of the candidate or voter regarding $d$ different \textit{issues} that are relevant for the election. In the social choice literature, preferences with this structure are often referred to as ($d$-)\e{Euclidean preferences} \cite{bogomolnaia2007euclidean,elkind2022preference}.

We consider a setting where only partial information about the preferences of voters is available. In such a setting, the exact preference order of a voter is unknown but assumed to come from a known space of possible preference orders. Each combination of possible preference orders is a possible voting profile that may result in different sets of winners (given a fixed voting rule). Natural computational tasks that arise in such scenarios ask about the \e{possible winners}  (who win in at least one possible profile) and the \e{necessary winners} (who win in every possible profile)~\citep{DBLP:conf/ijcai/Lang20}. A prominent manifestation of this idea is the seminal framework of \citet{Konczak2005VotingPW}, in which voter preferences are specified as partial orders and a possible profile is obtained by extending each partial order into a total preference order. A thorough picture of the complexity of the possible and necessary winner problems has been established in a series of studies~\citep{DBLP:journals/jcss/BetzlerD10,DBLP:journals/jair/XiaC11,DBLP:journals/ipl/BaumeisterR12}. For example, under every positional scoring rule in the setting of partial orders, the necessary winners can be found in polynomial time, yet it is NP-complete to decide whether a candidate is a possible winner (assuming a regularity condition), except for the tractable cases of the plurality and veto rules.

\begin{table*}[t]
\centering
\scalebox{0.89}{
\begin{tabular}{lll}
\toprule
{\textbf{Problem}}\phantom{XX} & Positional scoring rules &  Approval voting\\ 
\midrule
$\nwpar{d}$ & P [Thm.~\ref{thm:nwd}] & P for $d\leq 2$ [Thm.~\ref{thm:nwdApp}]  \\ 
\addlinespace[.8em]
$\pwpar{1}$ & \makecell[cl]{P for all two-valued rules [Thm.~\ref{thm:pw1Approval} and~\ref{thm:pwdPluVeto}], weighted veto [Thm.~\ref{res:pw1WeightedVeto}], \\ and $F(k,t)$ whenever $k > t$ [Thm.~\ref{res:pw1ThreeValued}]} & NP-c [Thm.~\ref{res:pw1App}]  \\
\addlinespace[.8em]
$\pwpar{d}$ & \makecell[cl]{P for plurality and veto [Thm.~\ref{thm:pwdPluVeto}];\\ 
NP-c for $k$-approval with $d \geq 2$ and $k \geq 3$ [Thm.~\ref{thm:pwdApproval}]} &  NP-c [Thm.~\ref{res:pw1App}] \phantom{XX}\\
\bottomrule
\end{tabular}
}
\caption{\label{tab:complexity} Complexity results for computing the necessary and possible winners in the studied models of uncertainty.}
\end{table*}

We study the complexity of the computational problems $\pwpar{d}$ and $\nwpar{d}$, where the goal is to find the possible and necessary winners, respectively, when we have incomplete information about voters' ideal points in spatial voting with $d$ dimensions. More precisely, instead of the ideal points,
we are given\,---\,for each voter and dimension\,---\,an interval of possible values for the voter's opinion. Hence, each voter is associated with a space of possible ideal points. 
Different points in this space may induce different preference orders over the candidates (whose locations are assumed to be known precisely). We thus get a mechanism for defining a space of possible total orders that is different\footnote{In \Cref{sec:model}, we show that the two settings are incomparable.} from the classical partial-order setting~\cite{Konczak2005VotingPW}. We refer to our setting as \e{partial spatial voting}. 

We first focus on the class of positional scoring rules and compare the computational complexity of the possible and necessary winner problems to the classic model of partial orders. We investigate the following questions:
{(1)}~Is the necessary-winner problem still tractable for all positional scoring rules? 
{(2)}~Is the possible-winner problem still tractable for plurality and veto? 
{(3)}~Are there positional scoring rules where the possible-winner problem is tractable for partial spatial voting but not for partial orders? We answer all three questions positively.
For some of our results, we uncover and exploit an interesting relationship between the possible-winner problem and \textit{scheduling}.  

We then consider spatial \textit{approval} voting with incomplete information. We provide an efficient algorithm for computing necessary winners in one and two dimensions and prove that computing possible winners is intractable, even for one dimension. 
Our results are summarized in \Cref{tab:complexity}.

\paragraph{Related work.}
Spatial voting in one dimension is intuitively similar to assuming \e{single-peaked preferences} \cite{Blac48a,Arro51a}. 
Yet, there are considerable differences, as single-peaked preferences do not impose any restrictions on the comparison between candidates on different sides of the peak. \citet{DBLP:conf/aaai/Walsh07} showed hardness results for possible and necessary winner questions 
under single-peaked (but not necessarily 1-Euclidean) preferences for STV and polynomial-time results for Condorcet-consistent rules.

\citet{bogomolnaia2007euclidean} showed that \e{every} (complete) preference profile can be represented in the spatial model, by choosing the dimension $d$ to be sufficiently large. 
Given a preference profile, it can be efficiently decided whether the profile can be represented as a one-dimensional spatial model~\cite{doignon1994polynomial,Knob10a,elkind2014recognizing}; for higher dimensions, the problem becomes intractable~\cite{peters2017recognising}. \citet{jamieson2011active} studied the problem of learning the ranking of candidates in spatial voting using pairwise comparisons. They established a bound on the number of possible rankings; we use this bound in \Cref{sec:necessaryWinners}. \citet{BGL+13a} and \citet{DBLP:conf/aaai/ImberIBK22} study the possible-winner and necessary-winner problems for approval voting in single-winner and multi-winner elections, respectively, using preference models similar to partial orders.

The problems considered in this paper also relate to manipulation and control problems that involve reasoning about a space of possibilities of profiles. \citet{DBLP:conf/atal/LuZROV19} study control where a party can select a subset of issues to focus on. \citet{DBLP:conf/uai/EstornellDEV20} study manipulation of spatial voting where the issues are weighted and a malicious attacker can change the weights. 
\citet{DBLP:conf/ijcai/0001EKV22} study manipulation where the adversary can change the position of a candidate.

\section{Preliminaries}\label{sec:preliminaries}

We first introduce the basic concepts and notation that we use throughout the paper. 

\subsection{Voting Profiles}
Let $C = \set{c_1, \dots, c_m}$ be a set of $m \ge 2$ \e{candidates} and $V = \set{v_1, \dots, v_n}$ a set of \e{voters}. A \e{ranking profile} 
$\R = (R_1, \dots, R_n)$ consists of $n$ linear orders over~$C$. Each $R_i$ represents the preference order of voter $v_i$. 
In the \e{spatial voting} setting~\citep{enelow1984spatial}, each candidate is associated with a $d$-dimensional vector corresponding to its positions (opinions) on $d$ issues, denoted by $\cvec{c}_i = \ip{c_{i, 1}, \dots, c_{i, d}} \in \real^d$.
For simplicity, we assume that all candidates have distinct positions, that is, there are no perfect clones.

A \e{spatial voting profile} $\T = (T_1, \dots, T_n)$ consists of a vector $T_j = \ip{T_{j, 1}, \dots, T_{j, d}}$ for each voter $v_j$, representing the voter's positions on the $d$ issues. Given a spatial profile $\T$, we can construct a ranking profile $\R_\T = (R_{T_1}, \dots, R_{T_n})$ where each voter $v_j$ ranks candidates in~$C$ according to their Euclidean distance $\norm{T_j-\cvec{c}_i}_2$ from~$v_j$. The closest candidate is ranked first, and the farthest is ranked last (position $m$) in $v_j$'s preferences. We break ties by a linear order over the candidates, which is given as part of the input for each voter. An example of a spatial voting profile and its associated ranking profile is illustrated in \Cref{fig:spatialVoting}. By a slight abuse of terminology, we may identify voters with their points in $\real^d$, and we use the terms \e{dimension} and \e{issue} interchangeably.

\begin{figure}
  \centering
  \scalebox{1}{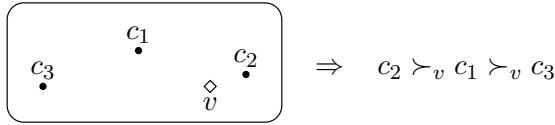}
  \caption{Example of a spatial voting profile with $d=2$: a set $C = \set{c_1, c_2, c_3}$ of candidates and a single voter $v$. \label{fig:spatialVoting}}
\end{figure}

\subsection{Voting Rules}\label{sec:votingRules}
A \e{voting rule} is a function that maps each ranking profile to a set of \e{winners}. A \e{positional scoring rule} $r$ is a series $\set{ \Vec{s}_m }_{m\ge 2}$ of $m$-dimensional score vectors $\Vec{s}_m = (\Vec{s}_m(1), \dots, \Vec{s}_m(m))$ of natural numbers, where $\Vec{s}_m(1) \geq \dots \geq \Vec{s}_m(m)$ and $\Vec{s}_m(1) > \Vec{s}_m(m)$. We assume that $\Vec{s}_m(j)$ is computable in polynomial time in $m$, and the scores in each $\Vec{s}_m$ are {mutually} co-prime (i.e., their greatest common divisor is one). Some examples of positional scoring rules include the \e{plurality} rule $(1, 0, \dots, 0)$, the \e{$k$-approval} rule $(1, \dots, 1, 0, \dots, 0)$ that begins with $k$ ones, the \e{veto} rule $(1, \dots, 1, 0)$, and the \e{$k$-veto} rule that ends with $k$ zeros.  A positional scoring rule is \e{pure} if every $\Vec{s}_{m+1}$ can be obtained from $\Vec{s}_m$ by inserting a score value at some position {(while satisfying $\Vec{s}_{m+1}(1) \geq \dots \geq \Vec{s}_{m+1}(m+1)$)}.

Given a ranking profile $\R = (R_1, \dots, R_n)$ a positional scoring rule $r$, the score $s_r(R_i, c)$ that the voter $v_i$ contributes to the candidate $c$ is $\Vec{s}_m(j)$, where $j$ is the position of $c$ in $R_i$. The score of $c$ in $\R$ is $s_r(\R, c) = \sum_{i=1}^n s_r(R_i, c)$, which we also denote as $s(\R, c)$ if $r$ is clear from context. A candidate $c$ is a \e{winner} if $s_r(\R, c)\geq s_r(\R, c')$ for all candidates $c'$. The set $r(\R)$ contains all winners. 

A positional scoring rule $r$ is \e{two-valued} if there are only two values in each $\Vec{s}_m$. For such rules, we assume, w.l.o.g., that $\Vec{s}_m$ consists only of zeros and ones, and hence is of the form $\Vec{s}_m = (1, \dots, 1, 0, \dots, 0)$. Thus, we can denote any two-valued rule as $k$-approval, where $k = k(m)$ may depend on $m$. For example, $(m-2)$-approval is the same as 2-veto.

For $k$-approval, we can naturally convert a ranking profile $\R = (R_1, \dots, R_n)$ to an \e{approval profile} $\A = (A_1, \dots, A_n)$, where each $A_i \subseteq C$ consists of the first $k$ candidates in the order $R_i$. In other words, $A_i$ denotes the set of candidates that the voter $v_i$ ``approves.'' The score $s(A_i, c)$ that the voter $v_i$ contributes to the candidate $c$ is one if $c \in A_i$ and zero otherwise. The winners then are the candidates with the maximal score $s(\A, c) = \sum_{i=1}^n s(A_i, c)$.

\subsection{Incomplete Profiles}
\def\calR{\mathcal{R}}
\def\uR{{\calR}}
\def\calA{\mathcal{A}}
\def\uA{\tilde{\calA}}
Throughout this paper, we study problems where voter preferences are incompletely specified, and we are interested in ``possible'' and ``necessary'' winners. 
Abstractly speaking, an incomplete voting profile is simply a set $\uR$
of ranking profiles. Given $\uR$, a candidate $c$ is called a \e{possible winner} w.r.t.~a voting rule $r$ if $c$ is a winner in at least one profile $\R\in\uR$, 
and a \e{necessary winner} w.r.t.~$r$ if $c$ is a winner in every profile $\R\in\uR$.
In contrast to possible winners, necessary winners may not exist. 

Incomplete profiles give rise to challenging computational problems when they are represented in a compact manner. For example, \citet{Konczak2005VotingPW} use a partial order over the candidates to represent the set of linear extensions, as we explain next. In the following section, we introduce another compact representation and compare it to their one.

\subsection{Partial Order Profiles}
A \e{partial order profile} $\P = (P_1, \dots, P_n)$ consists of $n$ partial orders (reflexive, anti-symmetric, transitive relations) on the set $C$ of candidates, where each $P_i$ represents the incomplete preferences of voter $v_i$. A \e{ranking completion} of $\P$ is a ranking profile $\R = (R_1, \dots, R_n)$ where each $R_i$ is a completion (i.e., linear extension) $P_i$.
As said above, a candidate $c$ is a necessary winner if $c$ is a winner in every ranking completion $\R$ of $\P$, and $c$ is a possible winner if there exists a ranking completion $\R$ of $\P$ where $c$ is a winner. 

For a positional scoring rule $r$, the decision problems  $\pw$ and $\nw$ (where ``po'' stands for ``partial order'') are those of determining, given a set $C$ of candidates,
a partial order profile $\P$ and a candidate $c \in C$, whether $c$ is a possible winner and a necessary winner, respectively. A classification of the complexity of these problems for positional scoring rules was established in a sequence of publications.

\begin{theorem}[\citet{DBLP:journals/jcss/BetzlerD10,DBLP:journals/jair/XiaC11,DBLP:journals/ipl/BaumeisterR12}]
  $\nw$ can be solved in polynomial time for every positional scoring rule. $\pw$ is solvable in polynomial time for plurality and veto; for all other pure positional scoring rules, $\pw$ is NP-complete.
\label{thm:classification}
\end{theorem}

\section{The Model of Partial Spatial Voting}\label{sec:model}

We introduce a model of incompleteness for spatial voting. A \e{partial spatial profile} $\P = (P_1, \dots, P_n)$ consists of a vector $P_j = \ip{[\ell_{j, 1}, u_{j, 1}], \dots, [\ell_{j, d}, u_{j, d}]}$ for every voter $v_j$. Each pair $[\ell_{j, i}, u_{j, i}]$ represents {a closed interval of possible values for the position of $v_j$ on issue $i$: $\ell_{j, i}$ is a lower bound, and $u_{j, i}$ is an upper bound.}
Note that the positions of the voters are incompletely specified, but those of the candidates are known precisely. Let $[n]$ denote $\set{1,\dots,n}$.
A spatial voting profile $\T = (T_1, \dots, T_n)$ is a \e{spatial completion} of $\P$ if $T_{j, i} \in [\ell_{j, i}, u_{j, i}]$ for every voter $v_j$ and issue $i \in [d]$.
We can then compute a ranking  profile $\R_\T$ as before.\footnote{{We can model the scenario where we do not know anything about voter $v_j$'s position regarding issue $i$ by setting $\ell_{j,i} = -\infty$ and $u_{j,i} = +\infty$. Our algorithms can handle this case efficiently through minor modifications.}}

We call a ranking profile $\R$ a \e{ranking completion} of $\P$ if there exists a spatial completion $\T$ such that $\R = \R_\T$. For $k$-approval, it will be useful to convert the ranking profile to an approval profile, as described in \Cref{sec:preliminaries}. We call an approval profile $\A$ an \e{approval completion} of $\P$ if there exists a spatial profile $\T$ such that $\R_\T$ is converted to $\A$.

Again, given a partial spatial profile $\P$, a candidate is a necessary winner if it is a winner in every ranking completion $\R$ of $\P$, and a possible winner if there exists a ranking completion $\R$ of $\P$ where it is a winner.  For a positional scoring rule $r$ and dimension $d$, we consider the decision problems where we are given a set $C$ of $d$-dimensional candidates, a partial spatial profile $\P$, and a candidate $c \in C$, and we need to determine whether $c$ is a possible or a necessary winner. We denote these two problems by $\pwpar{d}$ and $\nwpar{d}$, respectively. Note that the number of dimensions is fixed and not part of the input for the problem.

\subsection{Partial Spatial vs. Partial Order Profiles}
Before we move on, we make a note on the expressiveness of partial spatial voting compared to partial-order profiles. We say that a partial profile $\P$ (in one of the two models) can be \e{expressed} by the other model if there is a partial profile $\P'$ in the other model with the same set of ranking completions. In the case of full information, every (complete) profile can be expressed by a spatial profile with $d \leq \min\{m,n\}$ dimensions \cite{bogomolnaia2007euclidean}.

For every number $d$ of issues, we can easily come up with partial-order profiles (and even complete ranking profiles) that cannot be expressed as (partial) spatial profiles, simply by using the property that in spatial voting all voters must respect the positions of the candidates, where in partial orders each voter can have a completely different structure.
For example, if $d=1$ then preferences will be single-peaked.
Moreover, a partial order can have strictly more ranking completions than the upper bound of a partial spatial order. Indeed, while a partial order can have $\Omega(m!)$ completions, \citet{jamieson2011active} showed that a partial spatial voter has $\O(dm^{2d})$ completions (see \Cref{res:boundedRankingsHighDim}).

On the other hand, consider an instance with three candidates $C = \set{c_1, c_2, c_3}$, one-dimensional positions $\cvec{c}_1 = 1$, $\cvec{c}_2 = 2, \cvec{c}_3 = 3$, and a single voter with $P = [1, 3]$. The voter has four ranking completions: $(c_1, c_2, c_3)$, $(c_2, c_1, c_3)$, $(c_2, c_3, c_1)$, and $(c_3, c_2, c_1)$.
It is easy to verify that there is no partial order with this set of ranking completions.

We conclude that complexity results on possible or necessary winners for the partial-order model do not immediately imply results for the partial spatial model, and vice versa, since neither of the two models generalizes the other.

\section{Computing Necessary Winners}\label{sec:necessaryWinners}
In this section, we show that the necessary-winner problem can be solved in polynomial time, similarly to $\nw$ (as stated in \Cref{thm:classification}),  for every positional scoring rule and for every fixed number of dimensions.

\begin{theorem}
\label{thm:nwd}
Let $d \geq 1$ be fixed. $\nwpar{d}$ is solvable in polynomial time for every positional scoring rule.
\end{theorem}

The remainder of this section is devoted to proving \Cref{thm:nwd}.
To determine whether a candidate $c$ is a necessary winner for a given partial spatial profile $\P$, we use the same concept from the algorithm for $\nw$ given partial orders  \citep{DBLP:journals/jair/XiaC11}: a candidate $c$ is \e{not} a necessary winner if and only if there exists another candidate $c'$ and a ranking completion $\R$ where $s(\R, c') > s(\R, c)$. To this end, we iterate over every other candidate $c'$ and compute the maximal score difference $s(\R, c') - s(\R, c)$ among the ranking completions $\R$ of $\P$. Observe that it suffices to consider each voter $v_j \in V$ separately and compute the maximal score difference $s(R_j, c') - s(R_j, c)$ among the  ranking completions $R_j$ of $P_j$, since we can sum these values to obtain the maximal value of $s(\R, c') - s(\R, c)$. Then, $c$ is not a necessary winner if and only if the maximal score difference is positive for some candidate $c'$. 

The difference between our algorithm and the one for partial orders is in the way we compute the maximal score difference for each voter. We show that for a partial spatial voter, we can enumerate all ranking completions in polynomial time and compute the score difference in each ranking. This is impossible for partial orders, where the number of ranking completions can be exponential in $m$.

Next, we prove that we can enumerate the ranking completions of a single voter in polynomial time. We will use a geometric concept from the proof of \citet{jamieson2011active} of a polynomial bound on the number of possible complete rankings of a given vote.

\begin{lemma}[\citealp{jamieson2011active}] 
\label{res:boundedRankingsHighDim}
At most $\O(m^{2d})$ rankings can be constructed from spatial votes over a given sequence $\cvec{c}_1, \dots, \cvec{c}_m$ of $d$-dimensional candidates.
\end{lemma}

The proof of \Cref{res:boundedRankingsHighDim} relies on the following main idea.
Every pair $c,c'$ of  candidates corresponds to a $(d-1)$-dimensional hyperplane partitioning $\real^d$ into two $d$-faces (halfspaces): points closer to $c$, and points closer to $c'$. 
The set of all (at most) $\binom{m}{2} = m(m-1)/2$ hyperplanes then partitions $\real^d$ into a set $\Phi$ of regions, as illustrated in \Cref{fig:boundedRankings2D}. There is a one-to-one correspondence between these regions and the possible rankings: a face $\phi \in \Phi$ consists of exactly those points in $\real^d$ where the ranking of candidates in $C$ according to distance does not change, and no two of these faces correspond to the same ranking of candidates.

\Cref{res:boundedRankingsHighDim} implies that the number of ranking completions of a partial spatial vote 
is at most $\O(m^{2d})$. However, the bound itself is not enough for proving \Cref{thm:nwd}; we need to explicitly construct (and not just count) all feasible completions in order to parse them in the computation of the maximal score difference. Next, we explain how this is done. 

\begin{figure}[t]
  \centering
  \scalebox{1}{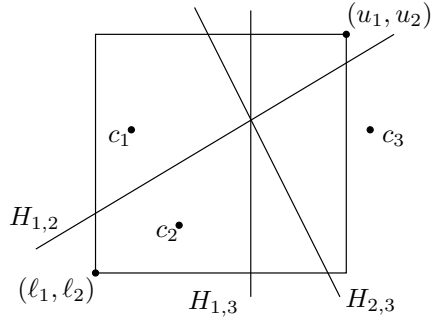}
  \caption{
  An illustration of the proof of  \Cref{res:boundedRankingsHighDim} for $d=2$ and $C = \set{c_1, c_2, c_3}$. A voter can be positioned at any point in the rectangle  $[\ell_1, u_1] \times [\ell_2, u_2]$. Each line $H_{i,j}$  partitions $\real^2$ into 2 regions: the points closer to $c_i$, and the points closer to $c_j$. 
  The top-left region\,---\,above $H_{1,2}$ and to the left of $H_{2,3}$ and $H_{1,3}$\,---\,corresponds to the possible positions of the voter where the preference ranking equals $c_1 \succ c_2 \succ c_3$.
  \label{fig:boundedRankings2D} 
  }
\end{figure}

\begin{algorithm}[t]
\caption{Enumerate Rankings}\label{alg:enumerateRankings}
\DontPrintSemicolon
\SetKwInOut{KwIn}{Input}
\SetKwInOut{KwOut}{Output}

\KwIn{$m$ candidates $C$ as points in $\mathbb{R}^d$ and the rectangle corresponding to $\P$ (via $2d$ linear inequalities)}
\KwOut{set of possible rankings of $C$ according to possible positions $\P$}

$H \leftarrow \emptyset$ \tcp*{set of hyperplanes}
$\Phi \leftarrow \emptyset$ \tcp*{set of $d$-faces}
$R \leftarrow \emptyset$ \tcp*{set of possible rankings}
\ForEach{\textup{distinct pair } $c,c' \in C$} {
    $h \leftarrow$ hyperplane corresponding to ``middle'' of $c$ and $c'$\;
    $H \leftarrow H \cup \{h\}$
}
$G \leftarrow$ incidence graph of $H$ \tcp*{can be obtained as in \cite{edelsbrunner1986constructing}}
\ForEach{\textup{node $x$ of $G$ corresp. to a $d$-face}} {
    $\phi \leftarrow \emptyset$ \tcp*{linear ineq.s bounding $d$-face $x$}
    \ForEach{\textup{node $y$ in $G$ corresp. to a $(d-1)$-face incident to $x$}}{
        $L \leftarrow $ linear inequalities corresponding to $y$\;
        $\phi \leftarrow \phi \cup L$
    }
    $\Phi \leftarrow \Phi \cup \phi$
}
\ForEach{$\phi \in \Phi$}{
    \If{$\phi \cap P$} {
        $w_\phi \leftarrow$ point in intersection of $\phi$ and $\P$\;
        $r \leftarrow C,$ sorted by distance to $w_\phi$\;
        $R \leftarrow R \cup \{r\}$
    }
}
\Return $R$
\end{algorithm}

\begin{lemma}
\label{res:iterateRankingsHighDim}
Let $C$ be a set of $m$ $d$-dimensional candidates and $P = \ip{[\ell_{1}, u_{1}], \ldots, [\ell_{d}, u_{d}]}$ a $d$-dimensional partial spatial voter. The set of ranking completions of $P$ can be enumerated in polynomial time.
\end{lemma}
\begin{proof}
The enumeration algorithm uses the geometric interpretation described above.
A pseudocode of the algorithm is given in \Cref{alg:enumerateRankings}.
Given the candidates $C$, we compute the corresponding set $H$ of at most $m(m-1) / 2$ hyperplanes. To represent the \e{arrangement} of these hyperplanes, i.e., of the geometric relation of the ($d$-)faces spanned by the hyperplanes, we construct an \e{incidence graph} $G(H)$. It consists of a node for each face of the arrangement, i.e., a node for each point, line(-segment), plane(-segment), etc., where two or more hyperplanes intersect. Furthermore, two nodes are connected by an edge if the corresponding faces are incident, i.e., one is contained in the other. Using an algorithm of \citet{edelsbrunner1986constructing}, $G(H)$ can be constructed in $\O(m^{2d})$ time.

We iterate over the nodes in $G(H)$ that correspond to $d$-faces of the arrangement. By \Cref{res:boundedRankingsHighDim}, there are at most $\O(m^{2d})$ such nodes. For each node $x$ be such a node. By considering all $(d-1)$-faces incident to $x$, as well as~$\P$, we represent the intersection of the $d$-face with $\P$ as a set of $\O(m^2 + d)$ linear inequalities with $d$ variables. We can then determine whether the $d$-face and $\P$ intersect, by checking the feasibility of a linear program with the aforementioned set of constraints, which  can be done in polynomial time. If there is a feasible solution (i.e., a point in the intersection), we compute the ranking from that point.
\end{proof}

This completes the proof of \Cref{thm:nwd}. Note that the running time of our algorithm is exponential in $d$, since we enumerate all ranking completions of a voter. {We can show} that for plurality and veto, one can find the necessary winners in time polynomial in $d$ (and in $n$ and $m$).
We do so by a reduction to $\nw$,
which is tractable~\citep{DBLP:journals/jcss/BetzlerD10}.
Hence, the necessary winners can be found in polynomial time for plurality and veto, even if $d$ is part of the input.

\begin{theorem}
    \label{thm:nwdPluVeto}
    For plurality and veto, $\nwpar{d}$ is solvable in time $\poly(n, m, d)$.
\end{theorem}
\begin{proof}
We show a reduction to $\nw$ under plurality and veto, respectively, in the model of partial orders, which is known to be solvable in polynomial time \citep{DBLP:journals/jcss/BetzlerD10}.
We start with the case of plurality. Let $C$ be a set of $d$-dimensional candidates and $\P = (P_1, \dots, P_n)$ a partial spatial profile. For every voter $v_j$, let $\Ctop^j \subseteq C$ be the set of candidates that are ranked in the first position in at least one ranking completion of $P_j$. (I.e., candidates that can receive a score of 1 from $v_j$.) Define a partial order $P_j' = (\Ctop^j \succ (C \setminus \Ctop^j))$. It is easy to verify that the ranking completions of $P_j$ and $P_j'$ result in the same possible scores. Therefore a candidate $c$ is a possible winner for the partial spatial profile $\P$ if and only if $c$ is a possible winner for the partial order profile $\P' = (P_1', \dots, P_n')$.

We now show that we can construct $\P'$ in polynomial time. Let $v_j$ be a voter with $P_j = \ip{(\ell_{j, 1}, u_{j, 1}), \dots, (\ell_{j, d}, u_{j, d})}$. Let $c_i \in C$. To test whether $c_i \in \Ctop^j$, we search for a completion $T_j = \ip{T_{j, 1}, \dots, T_{j, d}}$ of $P_j$ where $c_i$ is ranked first. Put differently, $\norm{T_j-\cvec{c}_i}_2 > \norm{T_j-\cvec{c}_t}_2$ for every other candidate $c_t \neq c_j$. Using the definition of Euclidean distance, we can rewrite this condition as 
\begin{equation} \label{eq:topRanked}
    2\sum_{k = 1}^d (c_{i, k} - c_{t, k}) T_{j, k} > \sum_{k = 1}^d c_{i, k}^2 - c_{t, k}^2
\end{equation}
We can therefore decide whether $c_i \in \Ctop^j$ by checking the feasibility of a linear program, where we search for a vector $T_j = \ip{T_{j, 1}, \dots, T_{j, d}}$ which satisfies $T_{j, k} \in [\ell_{j, k}, u_{j, k}]$ for every $k \in [d]$ and \Cref{eq:topRanked} for every candidate $c_t \neq c_i$. Since we have $d$ variables and $d+m-1$ inequalities, this can be done in time $\poly(n, m, d)$ \cite{Khac79a}, which completes the reduction. 

For veto we use the same idea with a small modification. For every voter $v_j$, let $\Cbottom^j \subseteq C$ be the set of candidates that are ranked last in at least one ranking completion of $P_j$, and define a partial order $P_j' = ((C \setminus \Cbottom^j) \succ \Cbottom^j)$. Each set $\Cbottom^j$ can be constructed in a very similar manner to $\Ctop^j$.
\end{proof}

\section{Computing Possible Winners}\label{sec:possible-winners}

We now turn to the problem of computing possible winners and show that, for some positional scoring rules, this problem is closely related to a scheduling problem. 

\subsection{The Case of a Single Dimension}
We start with the one-dimensional case ($d=1$) and study the complexity of $\pwpar{1}$. In this case, every candidate $c$ is associated with a single real value $\cvec{c}$. We assume without loss of generality that $\cvec{c}_1 < \dots < \cvec{c}_m$. A partial profile $\P$ consists of a pair $P_j = [\ell_j, u_j]$ for every voter $v_j$.
For partial orders, finding the possible winners is NP-complete except for plurality and veto (see \Cref{thm:classification}). In spatial voting, 
we are able to provide efficient algorithms for multiple well-studied classes of scoring rules. We begin our investigation with two-valued scoring rules,  
and then prove tractability for other classes of positional scoring rules.

\subsubsection{Two-Valued Rules}\label{sec:pw1TwoValued}
We begin with two-valued (one/zero) scoring rules. The simplest and most well-known rules in this class are plurality and veto.\footnote{We later show that for plurality and veto, $\pwpar{d}$ is solvable in polynomial time for every fixed $d \geq 1$. In particular, the tractability of $\pwpar{1}$ follows from both \Cref{thm:pw1Approval} and \Cref{thm:pwdPluVeto}.} 
We show that $\pwpar{1}$ is tractable for any two-valued rule. Recall that we denote a two-valued rule as $k$-approval for $k = k(m)$.

We introduce an alternative definition for partial spatial profiles for $k$-approval, in the case of a single dimension. Let $\P = (P_1, \dots, P_n)$ be a partial spatial profile where every voter $v_j$ is associated with a pair $P_j = [\ell_j, u_j]$. Since we assume $\cvec{c}_1 < \dots < \cvec{c}_m$, the set of candidates that $v_j$ possibly approves in a completion of $P_j$ is a sequence $(c_{i_j}, c_{i_j + 1}, \dots, c_{i_j + t})$ of consecutive candidates. 
(We can find this sequence in polynomial time using \Cref{res:iterateRankingsHighDim}.) 
In every completion, the candidates that $v_j$ approves form a substring of length $k$ of $(c_{i_j}, c_{i_j + 1}, \dots, c_{i_j + t})$; moreover, every such substring is the approval set of some completion.
Hence, we can define a partial spatial profile $\P = (P_1, \dots, P_n)$ for $k$-approval as follows. Each voter $v_j$ is associated with a sequence $P_j = (c_\ell, c_{\ell + 1}, \dots, c_u)$ of at least $k$ consecutive candidates.  In an approval completion $\A = (A_1, \dots, A_n)$ of $\P$, the set $A_j$
is a substring of length $k$ of $P_j$.
We then use $\A$ to compute the scores of the candidates and select the winners. 

With this definition, we solve $\pwpar{1}$ in polynomial time for $k$-approval. We use a reduction to scheduling with arrival times and deadlines, defined as follows.

\begin{definition}[Non-preemptive multi-machine scheduling with arrival times and deadlines]
\label{def:scheduling}
We are given a set $M = \set{M_1, \dots, M_t}$ of identical machines and a set $\J = \set{J_1, \dots, J_n}$ of $n$ jobs. Each job $J_j$ has an arrival time $a_j$, a deadline $d_j$, and processing time $p_j$. We assume that $a_j, d_j, p_j \in \natural$. A feasible schedule is a mapping $f \colon \J \rightarrow \real \times M$ that maps each $J_j \in \J$ to a pair $f(J_j) = (s_j, h_j)$ such that the following properties hold: 
\begin{enumerate}
    \item Every job is processed between its arrival time and deadline: $a_j \leq s_j \leq d_j-p_j$ for all $j \in [n]$.
    \item  Each machine runs at most one job at any time: 
    if $h_i = h_j$ for $i,j \in [n]$ with $i \neq j$ and $s_i \leq s_j$, then $s_j \geq s_i + p_i$.
\end{enumerate}
\end{definition}

Since the arrival times, deadlines, and processing times are all integers, we may assume w.l.o.g. that the starting time of every job in a feasible schedule is also an integer.\footnote{Otherwise, we can iterate over the jobs, sorted from the smallest starting time to the largest, and change the starting time $s_i$ of job $J_i$ to $\floor{s_i}$, without harming the feasibility.
} 
We now present the algorithm for $\pwpar{1}$.

\begin{theorem}
\label{thm:pw1Approval}
$\pwpar{1}$ is solvable in polynomial time under every two-valued positional scoring rule.
\end{theorem}
\begin{proof}
Let $k  = k(m)$. We are given a set $C$ of candidates, a candidate $c \in C$, and a partial spatial profile $\P = (P_1, \dots, P_n)$. We start by constructing another partial spatial profile $\P'$ with the following properties: \e{(i)} In $\P'$, each voter either necessarily approves $c$ or never approves $c$; and \e{(ii)} $c$ is a possible winner in $\P'$ if and only if $c$ is a possible winner in $\P$.

Let $V_c = \set{v_j \in V : c \in P_j}$ denote the voters who approve $c$ in at least one completion of $\P$. We define $\P' = (P_1', \dots, P_n')$ as follows. Let $v_j$ be a voter with $P_j = (c_\ell, c_{\ell+1}, \dots, c_u)$. If $v \notin V_c$ then $v$ never approves $c$, and we set $P_j' = P_j$. Otherwise, there exists $\ell \leq i \leq u$ for which $c_i = c$. Define $\ell' = \max\set{\ell, i-k+1}$ and $u' = \min\set{u, i+k-1}$. Then, $P_j' = (c_{\ell'}, c_{\ell' + 1}, \dots, c_{u'})$. Observe the following:
\begin{itemize}
    \item $\ell' \geq \ell$ and $u' \leq u$, hence $P_j'$ is a substring of $P_j$.
    \item $P_j'$ consists of at least $k$ candidates.
    \item Every substring of $P_j$ of length $k$ contains $c$. I.e., $c$ is necessarily approved by $v_j$ in $\P'$.
\end{itemize}

Observe that in $\P'$, the voters of $V_c$ necessarily approve $c$, and the voters of $V \setminus V_c$ never approve $c$, hence $s(\A, c) = |V_c|$ in every approval completion $\A$ of $\P'$. We show that $c$ is a possible winner in $\P$ w.r.t. $k$-approval if and only if $c$ is possible winner in $\P'$. One direction is trivial since every completion of $\P'$ is also a completion of $\P$. For the other direction, assume that there exists an approval completion $\A = (A_1, \dots, A_n)$ of $\P$ where $c$ is a winner. We construct a completion $\A' = (A_1', \dots, A_n')$ of $\P'$ as follows. Let $v_j$ be a voter. If $v_j \notin V_c$, or $v_j \in V_c$ and $c \in A_j$, then $A_j' = A_j$. Otherwise, $v_j \in V_c$ and $c \notin A_j$, i.e., $V$ does not approve $c$ in $A_j$. Let $A_j'$ be an arbitrary completion of $P_j'$. Note that $c \in A_j'$ by the construction of $P_j'$. We get that $s(A_j', c) = s(A_j, c) + 1$, and for every other candidate $c'$, $s(A_j', c') \leq s(A_j, c') + 1$, hence $c$ remains a winner when we change $A_j$ to $A_j'$. Overall, $c$ is a winner of $\A'$.

Next, we show a reduction from deciding whether $c$ is a possible winner in $\P'$ to multi-machine scheduling (from \Cref{def:scheduling}) where all jobs have the same processing time $k$. Deciding whether a feasible schedule exists in this setting is  solvable in polynomial time~\citep{vakhania2012}. In the reduction, the number of machines is $|V_c|$. For each voter $v_j \in V$ with $P'_j = (c_\ell, \dots, c_u)$, define a job $J_j$ with arrival time $\ell$, deadline $u+1$, and processing time $k$. 
We show that $c$ is a possible winner in $\P'$ w.r.t. $k$-approval if and only if there exists a feasible schedule of all the jobs.

Assume that there exists an approval completion $\A = (A_1, \dots, A_n)$ of $\P'$ where $c$ is a winner. Note that $s(\A, c) = |V_c|$ and $s(\A, c') \leq |V_c|$ for every other candidate $c' \in C$. We define a schedule as follows. For a voter $v_j$, let $A_j = \set{c_{i_j}, c_{i_j + 1} \dots, c_{i_j + k -1}}$ be the sequence of $k$ consecutive candidates that $v_j$ approves in $\A$. We schedule the job $J_j$ from time $i_j$ to time $i_j + k$. We show that this schedule is feasible for $|V_c|$ machines. By the definition of a completion, no job is started before its arrival and each job is completed by its deadline. Observe that for every $i \in [m]$ and $j \in [n]$, job $J_j$ is scheduled to run at time $[i, i+1]$ if and only if $v_j$ approves $c_i$ in $\A$. Since $s(\A, c') \leq |V_c|$ for every candidate $c' \in C$, we can deduce that any time at most $|V_c|$ jobs are scheduled, hence the schedule is feasible.

Conversely, assume there exists a feasible schedule (and recall our assumption that the starting times are integers). We construct a completion $\A = (A_1, \dots, A_n)$ of $\P$ as follows. For every voter $v_j$, let $[i_j, i_j + k]$ be the scheduled execution time of $J_j$ in the schedule. Define $A_j = \set{ c_{i_j}, c_{i_j + 1} \dots, c_{i_j + k -1} }$. Note that by the definition of the arrival time and the deadline, $A_j$ is a substring of $P_j'$ of length~$k$, hence each $A_j$ is indeed a completion of $P_j'$, and $\A$ is a completion of $\P'$. We show that $c$ is a winner of $\A$.
First, $s(\A, c) = |V_c|$ since $\A$ is a completion of $\P'$. Second, observe that for every $i \in [m]$ and $j \in [n]$, the voter $v_j$ approves $c_i$ in $\A$ if and only if the job $J_j$ is scheduled to run at time $[i, i+1]$. Since there are only $|V_c|$ machines, at any given time the number of jobs that are scheduled to run is at most $|V_c|$, which implies that $s(\A, c') \leq |V_c|$ for every candidate $c' \in C$. We deduce that $c$ is a winner in $\A$. 
\end{proof}

\subsubsection{Beyond Two-Valued Rules}
We consider two families of rules with more than two values. We refer to the first family as \e{weighted veto rules}. These rules are of the form $\Vec{s}_m = (\alpha, \dots, \alpha, \beta_1, \dots, \beta_k)$ for $\alpha > \beta_1 \geq \dots \geq \beta_k$ and $k < m/2$. The condition $k < m/2$ implies that each voter assigns the highest score $\alpha$ to more than half of the candidates. Moreover, for two positive integers $k$ and $t$, we denote by $F(k, t)$ the three-valued rule with scoring vector $\vec{s}_m = (2, \dots, 2, 1, \dots, 1, 0, \dots, 0)$ that begins with $k$ occurrences of two and ends with $t$ zeros. For example, the scoring vector for $F(2, 1)$ is $\vec{s}_m = (2, 2, 1, \dots, 1, 0)$.

\begin{theorem}
\label{res:pw1WeightedVeto}
$\pwpar{1}$ is solvable in polynomial time under every weighted veto rule.
\end{theorem}
\begin{proof}
Consider a weighted veto rule with $\Vec{s}_m = (\alpha, \dots, \alpha, \beta_1, \dots, \beta_k)$. Let $C$ be a set of 1-dimensional candidates and $\P$ a partial profile. We partition the candidates into three sets $\Cleft = \set{c_1, \dots, c_k}$, $\Cright = \set{c_{m-k+1}, \dots, c_{m}}$ and $\Cmiddle = C \setminus (\Cleft \cup \Cright)$. Note that $\Cmiddle \neq \emptyset$ since $k < m/2$. Each of the scores $\beta_1, \dots, \beta_k$ can only be assigned to the candidates of $\Cleft \cup \Cright$, and every voter assigns the score $\alpha$ to the candidates of $\Cmiddle$. Therefore for every ranking completion $\R$ of $\P$, every candidate $c \in \Cmiddle$ receives the maximal possible score $s(\R, c) = n \cdot \alpha$. We can deduce that all candidates of $\Cmiddle$ are possible winners. To complete the proof we describe how we can find the possible winners of $\Cleft \cup \Cright$ in polynomial time.

Let $c \in \Cleft \cup \Cright$. In order to be a winner in a ranking completion $\R$, $c$ must receive the maximal possible score $n \cdot \alpha$, otherwise it is defeated by the candidates of $\Cmiddle$. Hence, $c$ is a possible winner if and only if there exists a completion where every voter assigns the score $\alpha$ to $c$. We can therefore consider each voter separately, check if it can assign $\alpha$ to $c$ in some completion using \Cref{res:iterateRankingsHighDim}, and determine that $c$ is a possible winner if and only if that condition is satisfied for all voters.
\end{proof}

\begin{theorem}
\label{res:pw1ThreeValued}
$\pwpar{1}$ is solvable in polynomial time under $F(k,t)$ whenever $k > t$.
\end{theorem}
\begin{proof}
Let $C$ be a set of 1-dimensional candidates and $\P$ a partial profile. We denote the rule $F(k,t)$ by $r$ and $k$-approval by $r'$. Recall that for a ranking completion $\R$, we denote the score of $c$ in $\R$ w.r.t. a voting rule $r$ by $s_r(\R, c)$. As in the proof of \Cref{res:pw1WeightedVeto}, we partition the candidates to three sets $\Cleft = \set{c_1, \dots, c_t}$, $\Cright = \set{c_{m-t+1}, \dots, c_{m}}$ and $\Cmiddle = C \setminus (\Cleft \cup \Cright)$. 

We start by showing a connection between the scores of candidates in $\Cleft, \Cright$ and the scores of two specific candidates $c_{t+1}, c_{m-t} \in \Cmiddle$. Let $\R = (R_1, \dots, R_n)$ be a ranking completion of $\P$, let $v_j$ be a voter, and let $\cleft \in \Cleft$. Note that $v_j$ can only assign the score 0 to the candidates of $\Cleft \cup \Cright$, since it assigns 0 to the $t$ farthest candidates from its position, and $\Cleft \cup R$ are the first and last $t$ candidates on the line. Hence $s_r(R_j, c) \geq 1$ for every candidate $c \in \Cmiddle$, and in particular $s_r(R_j, c_{t+1}) \geq 1$.

We have three options for the score that $v_j$ assigns to $\cleft$. If the score is 0 then $s_r(R_j, \cleft) \leq s_r(R_j, c_{t+1}) + 1$, and if it is 1 then $s_r(R_j, \cleft) \leq s_r(R_j, c_{t+1})$. If $s_r(R_j, \cleft) = 2$ then we also have $s_r(R_j, c_{t+1}) = 2$ since $v_j$ assigns 2 to $k > t$ candidates, $\cleft$ is one of the $t$ left-most candidates, and $c_{t+1}$ is the $(t+1)$th candidate from the left. By summing the scores from all voters, we get the following:
\begin{align}
    s_r(\R, \cleft) \leq s(\R, c_{t+1}) - B(\R, \cleft) \label{eq:leftCandidateBound}
\end{align}
where $B(\R, c)$ is the number of voters who assign 0 to $c$ in $\R$. Similarly, for every $\cright \in \Cright$ we have
\begin{align}
    s_r(\R, \cright) \leq s(\R, c_{m-t}) - B(\R, \cright) \label{eq:rightCandidateBound}
\end{align}

We now show a connection between the scores of candidates in a completion w.r.t. $r$ and $r'$. Let $v_j$ be a voter and let $\cmiddle \in \Cmiddle$. Recall that $s_r(R_j, \cmiddle) \geq 1$. If $\cmiddle$ is among the top $k$ candidates in the ranking of $v_j$ then $s_r(R_j, \cmiddle) = 2$ and $s_{r'}(R_j, \cmiddle) = 1$. Otherwise, $\cmiddle$ is not among the top $k$ candidates in the ranking of $v_j$, which implies $s_r(R_j, \cmiddle) = 1$ and $s_{r'}(R_j, \cmiddle) = 0$. In both cases we get $s_r(R_j, \cmiddle) = s_{r'}(R_j, \cmiddle) + 1$, and overall
\begin{align}
     s_r(\R, \cmiddle) = s_{r'}(\R, \cmiddle) + n \label{eq:middleCandidateTransform}
\end{align}
For a candidate $\cleft \in \Cleft$ we have the same relation between the scores under $r$ and $r'$, unless $s_r(R_j, \cleft) = 0$ which implies $s_{r'}(R_j, \cleft) = s_r(R_j, \cleft) = 0$. We can apply the same argument for every $\cright \in \Cright$, and obtain the following.
\begin{align}
    s_r(\R, \cleft) &= s_{r'}(\R, \cleft) + n - B(\R, \cleft) \label{eq:leftCandidateTransform} \\
    s_r(\R, \cright) &= s_{r'}(\R, \cright) + n - B(\R, \cright) \label{eq:rightCandidateTransform}
\end{align}

We now present the algorithm to determine whether a candidate $c$ is a winner. We use a different procedure for each set of candidates $\Cleft, \Cmiddle, \Cright$. Let $c \in \Cmiddle$. We show that for every ranking completion $\R$ of $\P$, $c$ is a winner of $\R$ w.r.t. $r$ if and only if $c$ is a winner of $\R$ w.r.t. $r'$. Finding the possible winners under $r'$ is covered by \Cref{thm:pw1Approval}, hence we get that we can decide whether $c$ is a possible winner w.r.t. $r$ in polynomial time.

Assume that $c$ is a winner of $\R$ w.r.t. $r$, i.e., $s_r(\R, c) \geq s_r(\R, c')$ for every other candidate $c'$. Let $\cmiddle \in \Cmiddle$, by \Cref{eq:middleCandidateTransform} we get $s_{r'}(\R, c) \geq s_{r'}(\R, \cmiddle)$.
For every $\cleft \in \Cleft$, we apply Equations \ref{eq:leftCandidateBound}
and \ref{eq:leftCandidateTransform}.
\begin{align*}
    s_{r'}(\R, c) &= s_r(\R, c) - n \geq s_r(\R, c_{t+1}) - n \geq (s_r(\R, \cleft)) + B(\R, \cleft)) - n = s_{r'}(\R, c_{t+1})
\end{align*}
For $\cright \in \Cright$ we apply Equations \ref{eq:rightCandidateBound} and \ref{eq:rightCandidateTransform} in the same manner. We can deduce that $c$ is a winner of $\R$ w.r.t. $r'$.

Now, assume that $c$ is a winner of $\R$ w.r.t. $r'$. For every $\cmiddle \in \Cmiddle$ we can use \Cref{eq:middleCandidateTransform} again to show that $s_r(\R, c) \geq s_r(\R, \cmiddle)$. For $\cleft \in \Cleft$ we use Equations \ref{eq:middleCandidateTransform} and \ref{eq:leftCandidateTransform}.
\begin{align*}
    s_r(\R, c) &= s_{r'}(\R, c) + n \geq s_{r'}(\R, \cleft) + n \geq  s_{r'}(\R, \cleft) + n - B(\R, \cleft) = s_r(\R, \cleft)
\end{align*}
For $\cright \in \Cright$ we apply Equations \ref{eq:middleCandidateTransform} and \ref{eq:rightCandidateTransform} in the same manner. We can deduce that $c$ is a winner of $\R$ w.r.t. $r$. This completes the proof for the algorithm of $\pwpar{1}$ in the case that $c \in \Cmiddle$. 

We now show an algorithm for $c \in \Cleft$, and the case of $c \in \Cright$ is similar. In a completion $\R$, if $B(\R, c) > 0$ then $c$ is not a winner of $\R$ since $ s_r(\R, c) < s_r(\R, c_{t+1})$ by \Cref{eq:leftCandidateBound}. We define another partial profile $\P'$ where voters never assign 0 to $c$ and get that $c$ is a possible winner of $\P$ if and only if it is a possible winner of $\P'$. Note that $\P'$ can be easily constructed by inspecting the ranking completions of each voter $v_j$ and modifying the values of $P_j = (\ell_j, u_j)$ accordingly.

In every completion $\R$ of $\P'$ we have $s_r(\R, c) = s_{r'}(\R, c) + n$, since we can use \Cref{eq:leftCandidateTransform} with $B(\R, c) = 0$. By the same arguments that we had for the case of $c \in \Cmiddle$, we can show that~$c$ is a possible winner of $\P'$ w.r.t. $r$ if and only if it is a possible winner of $\P'$ w.r.t. $r'$.
\end{proof}

\subsection{The Case of Multiple Dimensions}
\label{sec:pw-d>1}

For $d > 1$, we show that the tractable cases of possible winners for the partial orders model, namely plurality and veto, are also tractable for the spatial model. 

\begin{theorem}\label{thm:pwdPluVeto}
For plurality and veto, $\pwpar{d}$ is solvable in time $\poly(n, m, d)$.
\end{theorem}
To prove \Cref{thm:pwdPluVeto}, we reduce $\pwpar{d}$ to $\pw$, which is solvable in polynomial time~\citep{DBLP:journals/jcss/BetzlerD10}, 
similar to the proof of \Cref{thm:nwdPluVeto}.
We get again a running time that is polynomial in $n$, $m$ and $d$. Hence, the possible winners can be found in polynomial time for plurality and veto even if $d$ is part of the input. For $k$-approval, the problem becomes intractable for $d \geq 2$, and every $k \geq 3$.

\begin{figure}
  \centering
  \scalebox{1}{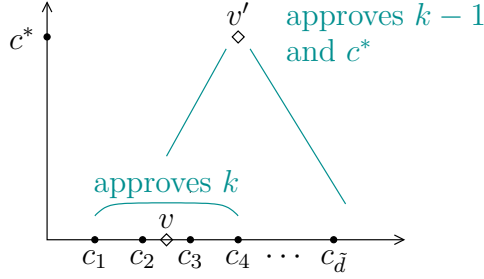}
  \caption{An example of two voters in a completion of the partial profile $\P$ from the proof of \Cref{thm:pwdApproval}. The voter $v$ represents a job of length $k$, and approves the $k$ candidates closest to it among $c_1, \dots, c_{\dtilde}$. The voter $v'$ represents a job of length $k-1$ and approves $c^*$ and the $k-1$ candidates closest to it among $c_1, \dots, c_{\dtilde}$. \label{fig:pwdApproval}}
\end{figure}

\begin{theorem}
\label{thm:pwdApproval}
    Let $d \geq 2$ and $k \geq 3$ be fixed. $\pwpar{d}$ is NP-complete for $k$-approval.
\end{theorem}
\begin{proof}
    We focus on $d=2$, since this is a special case of any {$d > 2$}. We show a reduction from non-preemptive scheduling with arrival times and deadlines (from \Cref{def:scheduling}) where we have a single machine and every processing time satisfies $p_j \in \set{k, k-1}$. Deciding whether a feasible schedule exists is strongly NP-complete for every $k>2$~\citep{DBLP:journals/corr/ElffersW14};
    then, we may assume that the maximal deadline $d_{\max}$ is polynomial in the number of jobs~$n$. We also assume that the minimal arrival time is 1.
    Let $\J_{k}$ and $\J_{k-1}$ be the sets of jobs with processing times $k$ and $k-1$, respectively. Each job $J_j$ {has} an arrival time $a_j$ and a deadline $d_j$. Let $\dtilde$ be the smallest multiple of $k$ that is greater or equal to $d_{\max}-1$.
    
    In the reduction, illustrated in \Cref{fig:pwdApproval}, the candidates are $C = \set{c^*, c_1, \dots, c_{\dtilde}}$. The positions are $\cvec{c}^* = \ip{0, 3\dtilde}$ for $c^*$ and $\cvec{c}_i = \ip{i+0.5, 0}$ for every $c_i$.
    For every partial voter that we define, the interval for the first position $(\ell_{j,1}, u_{j,1})$ always satisfies $0 \leq \ell_{j,1} \leq u_{j,1} \leq \dtilde$, and the position on the second issue is always one of two values, either 0 or $3\dtilde$.
    Observe that for a partial voter $v$ with $P = \ip{(\ell, u), 0}$ such that $0 \leq \ell \leq u \leq \dtilde$, the set of candidates that $v$ approves among the completions of $P$ is a sequence of consecutive candidates $(c_{i}, c_{i + 1}, \dots, c_{i + t})$. Moreover, in every completion, the candidates that $v$ approves are a substring of length $k$ of that sequence. If we have $P = \ip{(\ell, u), 3\dtilde}$ then $v$ necessarily approves $c^*$. Besides $c^*$, the set of candidates that $v$ approves among the completions of $P$ is again a sequence of consecutive candidates, and in each completion, $v$ approves a substring of length $k-1$ of that sequence. 
    We can therefore use an alternative definition for partial voters as in \Cref{sec:pw1TwoValued}. For each voter we specify the value on the second issue, and a sequence of candidates that the voter approves among the completions (not including $c^*$, since each voter either necessarily approves it or never approves it). Hence we denote $P = \ip{(c_i, \dots, c_{i+t}), y}$ where $(c_i, \dots, c_{i+t})$ is a sequence of consecutive candidates and $y \in \{ 0, 3\dtilde \}$. An approval completion is defined as before (for each voter, we specify the candidates that it approves).

    The partial profile $\P = \P^1 \circ \P^2 \circ \P^3$, consists of three parts that we describe next. In the first part, $\P^1$, for every job $J_j \in \J_{k}$ we introduce a voter $v_j$ with $P_j = \ip{(c_{a_j}, \dots, c_{d_j-1}), 0}$. In the second part, $\P^2$, for every job $J_j \in \J_{k-1}$ we introduce a voter $v_j$ with 
    $P_j = \ip{(c_{a_j}, \dots, c_{d_j-1}), 3\dtilde}$. The third part $\P^3$ consists of $(|\J_{k-1}|-1) \cdot \dtilde / k$ voters without uncertainty such that every candidate among $c_1, \dots, c_{\dtilde}$ is approved by exactly $|\J_{k-1}|-1$ voters. Specifically, $|\J_{k-1}|-1$ voters approve $c_1, \dots, c_k$, then $|\J_{k-1}|-1$ voters approve $c_{k+1}, \dots, c_{2k}$ and so on.

    We now analyze the possible scores in different completions. Let $\A = \A^1 \circ \A^2 \circ \A^3$ be an approval completion of $\P$. For every voter $v_j$ in $\A^1$, as we stated in the discussion on the alternative model of partial votes, $c^*$ is never approved by $v_j$, hence  $s(\A^1, c^*) = 0$. In contrast, every voter $v_j$ in $\A^2$ approves $c^*$, which implies $s(\A^2, c^*) = |\J_{k-1}|$. For the third part $\A^3$ we get $s(\A^3, c^*) = 0$ and $s(\A^3, c_i) = |\J_{k-1}|-1$ for every candidate $c_i$. Overall, the score of $c^*$ is $s(\A, c^*) = |\J_{k-1}|$.
    
    We show that $c^*$ is a possible winner of $\P$ if and only if there exists a feasible schedule. Let $\A$ be an approval completion of $\P$ where $c^*$ a winner. By our analysis of the scores in different completions, for every candidate $c_i$ we have $s(\A, c_i) \leq s(\A, c^*) = |\J_{k-1}|$. Since $s(\A^3, c_i) = |\J_{k-1}|-1$ we get $s(\A^1 \circ \A^2, c_i) \leq 1$, i.e., at most one voter from $\A^1 \circ \A^2$ approves $c_i$. We construct a schedule as follows. For every job $J_j \in \J_k$, let $(c_{i_j}, c_{i_j+1}, \dots, c_{i_j+k-1})$ be the $k$ candidates that $v_j$ approves in $A_j \in \A^1$. We schedule $J_j$ to start at time $i_j$. For every job $J_j \in \J_{k-1}$, let $(c_{i_j}, c_{i_j+1}, \dots, c_{i_j+k-2})$ be the $k-1$ candidates that $v_j$ approves in $A_j \in \A^2$ other than $c^*$. We again schedule $J_j$ to start at time $i_j$. 
    We show that this is a feasible schedule. By the definition of the partial votes, every job is processed between its arrival time and its deadline. Let $c_i \in C$. Observe that a job $J_j$ is scheduled to run at time $i$ if and only if the voter $v_j$ approves $c_i$ in $\A^1 \circ \A^2$. Since At most one voter from $\A^1 \circ \A^2$ approves $c_i$, we can deduce that at most one job is scheduled to run at time $i$, therefore we never schedule two jobs at the same time and the schedule is feasible.

    We now prove the other direction. Assume there exists a feasible schedule (recall we can assume the starting times are all integers). We construct a completion $\A = \A^1 \circ \A^2 \circ \A^3$ of $\P$. For every job $J \in \J_k$, let $i_j$ be the scheduled starting time, we define $A_j =  \{c_{i_j}, c_{i_j+1}, \dots, c_{i_j+k-1}\}$.  For every job $J \in \J_{k-1}$, let $i_j$ be the scheduled starting time, we define $A_j = \{c^*, c_{i_j}, c_{i_j+1}, \dots, c_{i_j+k-2}\}$. Note that by the definition of the arrival times,  deadlines, and the two types of jobs, each $A_j$ is a valid completion of $P_j$, and $\A$ is a completion of $\P$. We show that $c$ is a winner of $\A$.
    By our analysis regarding the scores in different completions, it is sufficient to show that $s(\A^1 \circ \A^2, c_i) \leq 1$ for every candidate $c_i$. For every candidate $c_i$ and voter $v_j$, $v_j$ approves $c_i$ in $\A^1 \circ \A^2$ if and only if the job $J_j$ is scheduled to run at time $i$. Since at most one job is scheduled to run at any time, we get that at most one voter of $\A^1 \circ \A^2$ approves $c_i$. This completes the proof.
\end{proof}

\section{Partial Spatial Approval Voting}\label{sec:ABC}

We now turn our attention to approval voting (AV), where each voter partitions the set of candidates into ``approved'' and ``unapproved'' candidates and the candidate with the highest number of approvals is selected \citep{BrFi83b}. In contrast to the positional scoring rule $k$-approval, the number of approved candidates is not fixed under AV. 
In the spatial framework, it is often assumed that each voter $v_j$ is associated with an \e{approval radius} $\rho_j \in \real$ and approves all candidates that are at a distance of at most $\rho_j$ from voter $v_j$'s position~\cite[e.g.,][]{GBSK21a}. 

More formally, a \e{partial spatial approval profile} consists of a partial spatial profile $\P$ and a radius $\rho_j \in \real$ for every voter $v_j$.
Each spatial completion $\T = (T_1, \dots, T_n)$
of $\P$ gives rise to an approval profile $\A_\T = (A_{T_1}, \dots, A_{T_n})$, where $A_{T_j} = \{c \in C: \norm{T_j-\cvec{c}_i}_2 \le \rho_j  \}$ is the approval set of voter $v_j$. We call $\A_\T$ an \textit{approval completion} of $T$ (keeping in mind that, in contrast to the approval completions in \Cref{sec:model}, the approval sets can have different sizes).
The definitions of possible and necessary winners and their respective decision problems remain unchanged.

First, for $d \leq 2$ dimensions, we show an upper bound on the number of different approval completions a single voter can have, and how to enumerate them in polynomial time.
This is crucial for our results for necessary-winner problems.
The regions in $\mathbb{R}^d$ where a voter has the same approval completion are bounded by (intersections of) $d$-dimensional \textit{spheres} (rather than hyperplanes as in \Cref{sec:necessaryWinners}).
We consider the connection between a partial spatial approval profile, for a single voter, and an arrangement of $d$-spheres (with the $d$-dimensional rectangle of the voter).
Instead of assuming the voter has a radius $\rho$ around its position and approves all candidates inside this sphere, we can equivalently assume that we have a sphere with radius $\rho$ around each \e{candidate's} position and the voter approves all candidates whose sphere contains the voter's position. See \Cref{fig:spatialApproval} for a graphical depiction in two dimensions.

\begin{figure}
\centering
\scalebox{1}{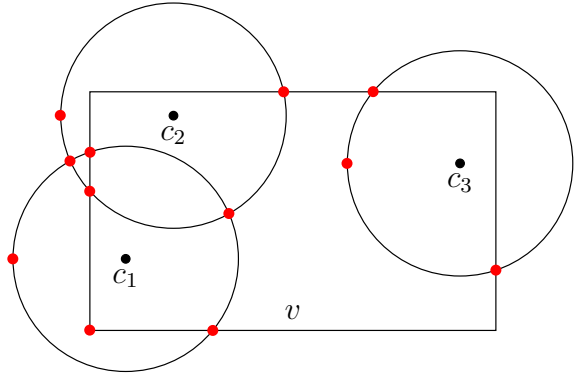}
\caption{\label{fig:spatialApproval} Example of an arrangement for $d=2$, candidate set $\set{c_1,c_2,c_3}$, and a voter $v$.
The possible approval sets of $v$ are $\emptyset, \{c_1\}, \{c_2\}, \{c_3\},$ and $\{c_1, c_2\}$, depending on the actual position of $v$ inside the rectangle. The red points are all event points from the sweep line algorithm in \Cref{res:iterateApprCompletions}.}
\end{figure}

Now, recall that the voter's possible positions are described by a $d$-dimensional rectangle. Thus, if this rectangle intersects a sphere associated with a candidate $c$, then there is a completion where the voter approves $c$. Moreover, if there is a point inside the rectangle that is contained in the spheres of, say, the candidates $c_1$, $c_2$, and $c_3$, and in none of the other spheres, then there is a completion of the profile where the voter's approval set is exactly $\set{c_1,c_2,c_3}$: the voter's approval set is $\set{c_1,c_2,c_3}$ in exactly the completions where the position of the voter is in the intersection of the rectangle of the voter and the three spheres around $c_1$, $c_2$ and $c_3$. Thus, the possible completions of the voter's approval profile correspond to all (maximal) regions of the arrangement of $m$ $d$-spheres with radius $\rho$ (each around the position of one of the candidates), and the $d$-dimensional rectangle described by $\ip{[\ell_{j, 1}, u_{j, 1}], \dots, [\ell_{j, d}, u_{j, d}]}$ (where we only consider those regions that lie inside the rectangle).

\begin{theorem}
\label{res:iterateApprCompletions}
Fix $d\le 2$ and let $C$ be a set of $d$-dimensional candidates and $P$ 
a single $d$-dimensional partial vote with radius $\rho \in \real_{\geq 0}$. Then, the set of approval completions of $P$ can be enumerated in polynomial time.
\end{theorem}
\begin{proof}
We provide proofs for $d=1$ and $d=2$ separately but use the same idea.
To enumerate the regions of the arrangement as described above, we use a standard tool in computational geometry called a \e{sweep line} algorithm (e.g., see \citealp{halperin2017arrangements}). The idea is to go over the arrangement from left to right, focusing on \e{event points} where something in the arrangement changes.

For $d=1$, the spheres and the rectangle all degenerate to intervals in $\real$, we denote the intervals corresponding to the candidates by \e{candidate-intervals}. We construct a set $X$ of binary arrays of length $m$, which represents all approval completions of $P$. We start with $X = \emptyset$ and ``sweep'' the interval $[\ell,u]$ of the partial vote from left to right and stop at the following event points (note that by only considering the interval $[\ell,u]$ we will not get all regions but only those that intersect the interval of the voter\,---\,which are actually the only ones we are interested in).
An event point here is each left and right border of a candidate-interval as well as the points $\ell$ and $u$. Note that since all candidates' positions are unique but all candidates use the same radius $\rho$, no two left or two right endpoints of the candidate-intervals will lie on top of each other.
If we reach such an event point $p$ we measure the distance $d_i$ from this point to each candidate $c_i \in C$.

For each event point $p$, we costruct an array $x \in \set{0,1,\diamond}^m$, where $\diamond$ is a placeholder that we will replace by 0 or 1 in a later step. For each candidate $c_i$,
\[ x_i = \begin{cases}
1 &\text{if $d_i < \rho$}\\
0 &\text{if $d_i > \rho$}\\
\diamond &\text{otherwise ($d_i = \rho$)}
\end{cases}\]
Observe that if $d_i < \rho$ then $p$ is inside $c_i$'s candidate interval; if $d_i > \rho$ then $p$ lies outside $c_i$'s candidate interval; and if $d_i = \rho$ then $p$ lies on the boundary of $c_i$'s candidate interval.

If there is only one $\diamond$ in $x$ we are at an event point where one of intervals starts or ends. In this case we add two arrays to $X$, one where $\diamond$ is exchanged for a 0 and one where it is exchanged for a 1 corresponding to the two regions that touch at this point. It might be the case that one of the arrays is already in $X$, in which case we simply do not include it a second time. If there is more than one $\diamond$ in $x$ then $p$ is an event point where one candidate-interval ends and exactly one other starts, hence we will have two $\diamond$ in $x$, let $i, j$ be their indices. (More than two $\diamond$ is not possible since we assume that no two candidates have the exact same position.) We then add three new arrays to the list of regions: one where $x_i = 0, x_j = 1$, one where $x_i = 1, x_j = 0$, and one where $x_i = x_j = 1$. There is no region that would correspond to the case where we set $x_i = x_j = 0$ as the intervals (which are assumed to be closed) touch at $p$. Once we construct $X$ with the arrays of all event points (i.e., we reach the last event point at $u$) the algorithm stops.

Since we have at most $2m+2$ event points and at each point we perform $m$ distance checks the algorithms runs in time $\O(m^2)$.
To see that the algorithm introduces an array for each region inside $[\ell,u]$ of the arrangement in $\real$, fix any such region. The region has a leftmost point, i.e.\ a point with the lowest coordinate. This is either $\ell$ or the beginning of a candidate-interval. In either case, the leftmost point was an event point and we introduce all regions that meet at an event point.

We now consider the case that $d=2$. The overall approach is exactly the same as for $d=1$, we only have to take care of the definitions of event points. Once we establish these we consider them from left to right (i.e.\ in order from lowest to highest $x$-coordinate, breaking ties by their $y$-coordinates) and construct a set $X$ of binary vectors of length $m+1$. In each vector, the $i$'th entry is $1$ if the region lies inside the sphere of $c_i$, and $0$ otherwise. The $(m+1)$st entry is $1$ if the region is inside the rectangle.
The event points this time are all $m$ leftmost points of the circles corresponding to the candidates and the one leftmost (i.e.\ the  bottom left) point of the rectangle corresponding to the voter. Additionally, every intersection point of the circles among themselves and with the rectangle constitutes an event point. See \Cref{fig:spatialApproval} for an example. Since two circles intersect in at most 2 points and a rectangle and a circle intersect in at most 8 points the number of event points is at most $2\binom{m}{2} + 9m + 1 = \mathcal{O}(m^2)$, and these points can be computed in polynomial time.

Upon reaching an event point $p$, we again measure for each candidate the distance from $p$ to this candidate and fill $x \in \set{0,1,\diamond}^m$ as described above. For each candidate $c_i$ we fill $x_i$ as in the algorithm for $d=1$. For $x_{m+1}$ we check via whether the point lies inside the voter's rectangle or not and assign 0,1, or $\diamond$. In the following we will denote the $m$ spheres and the rectangle together simply by \e{objects} as they will behave mostly identically for the remainder of the algorithm.

For each pair of objects with a $\diamond$ we check whether they touch at $p$ or truly intersect each other. This can be done as a preprocessing step for every pair of objects while computing the intersection points: two circles only touch each other if and only if they have only one intersection point (since they have the same radius); for a circle and the rectangle this is only slightly more involved and can be done in the quadratic timeframe of finding all intersection points.
If the objects are intersecting, then we introduce all regions coming from all combinations of 0 and 1 for the different $\diamond$'s. If they are only touching, we do not introduce an array where we set both $\diamond$'s to 0. In the case that there are more than two objects intersecting and/or touching at $p$, we consider them together. Again, if an array is already in $X$, we do not need to add it again.

The algorithm stops after handling all event points. Since at each of the $\O{O}(m^2)$ many event points we essentially only have to perform $m$ distance operations (and 4 linear equation checks for the rectangle) the algorithm runs in polynomial time.
We now argue again why each region is introduced.
First, note that when considering an event point $p$ we obtain every region that $p$ touches by the above procedure of using the placeholder $\diamond$. Now consider any region. Its leftmost point is either an intersection point of two or more objects or the leftmost point of an object. Thus, the leftmost point of the region is an event point and the region gets introduced. 
\end{proof}

We can use the above result to find the necessary winners for AV if $d \leq 2$, similarly to the proof of \Cref{thm:nwd}. Note that we run the sweep line algorithm described above once for each voter $v_j$, using that voters radius $\rho_j$.\footnote{Note that every two $d$-spheres intersect in a $(d-1)$-sphere. For $d \geq 3$, it is thus not clear how to define the corresponding event points for a sweep plane algorithm in $d$ dimensions.} 

\begin{theorem}
\label{thm:nwdApp}
Let $d \leq 2$. $\nwpar{d}$ is solvable in polynomial time for AV.
\end{theorem}

On the other hand, we show that the possible-winner problem is hard, even for $d=1$. This is the first negative result we have for partial spatial voting (using approval or ranked preferences) in the one-dimensional case.

\begin{theorem}
\label{res:pw1App}
$\pwpar{d}$ is NP-complete for any $d \geq 1$ for AV.
\end{theorem}

\begin{proof}
    We show a reduction from non-preemptive scheduling with arrival times and deadlines (see \Cref{def:scheduling}), where we have a single machine. Deciding whether a feasible schedule exists is strongly NP-complete 
    \citep{GJ79}.
    As in the proof of \Cref{thm:pwdApproval}, we may assume that the maximum deadline $d_{\max}$ is polynomial in the input size.
    Given $n$ jobs jobs $\J = \set{J_1, \dots, J_n}$,we construct an instance of $\pwpar{1}$ for apporval voting with $2 + d_{\max}$ candidates and $n+1$ voters such that the given scheduling instance is a \e{yes}-instance if and only if the $\pwpar{1}$ instance is a yes instance. Since the real line can be embedded in any $\mathbb{R}^d$, $d\geq 2$, the result immediately translates to higher dimensions.

    The candidates are $C = \set{c^*, c_0, c_1, \dots, c_{d_{\max}}}$, each $c_i$ is positioned at $\cvec{c_i} = i$ on the real line and $\cvec{c^*} = -1$. There is one voter $v^*$ who without uncertainty only approves of $c^*$. This is the candidate for which we will later try to solve the possible winner question. For each job $J_j \in \J$ we have a voter $v_j$ that has radius $\rho_j = \frac{1}{2} \cdot p_j$ and is positioned such that $v_j$ can approve any $p_j$-length interval of the candidates $c_{a_j}, \ldots, c_{d_j}$. In particular, for some $\varepsilon > 0$ we set 
    \begin{align*}
    \ell_j = a_j + \frac{p_j}{2} - \varepsilon,\, \text{ and } u_j = d_j - \frac{p_j}{2} + \varepsilon \,.
    \end{align*}
    Note that $v_j$ can also approve intervals of length $p_j +1$ of $c_{a_j}, \ldots, c_{d_j}$ according to this construction. Nevertheless, we can assume $v_j$ only ever approves intervals of length $p_j$ because if $c^*$ is a possible winner in any completion where $v_j$ approves of $p_j + 1$ candidates, $c^*$ is also a winner in the completion where $v_j$ approves of a $p_j$-sized sub-interval of those candidates instead. Thus, we can assume that $v_j$ only approves of some $p_j$-sized interval of $c_{a_j}, \ldots, c_{d_j}$.

    We now argue why $c^*$ is a possible winner if and only if the scheduling instance is a \e{yes}-instance.
    First note that $c^*$ has a score of 1 in every completion and thus is a possible winner if and only if no other candidate has a score greater or equal to 2. This is exactly the case if no two voters approve the same candidate, which in turn is equivalent to all tasks being scheduled (on a single machine) without overlap. This is the condition for the sequencing instance to be a \e{yes}-instance.
    Completeness follows by noting that a single completion in which the candidate of interest is a winner functions as a witness for a \e{yes}-instance.
\end{proof}

\subsection{Approval-based committee voting.}
We now briefly discuss the impact of these results on approval-based \e{multi-winner} elections.
In approval-based committee (ABC) voting \citep{LaSk22a}, a \e{committee} (i.e., subset of candidates) of a fixed size $k$ needs to be selected based on approval ballots.
Adapting the notions of \citet{DBLP:conf/aaai/ImberIBK22} to our setting, a set of candidates $W$ is a \e{possible committee} if there is a completion of the partial spatial approval profile where $W$ is selected, and a \e{necessary committee} if $W$ wins in every completion. We parametrize both problems by the dimension $d$ and the committee size $k$.

The hardness of possible winners of the single-winner AV rule (\Cref{res:pw1App}) can be extended to ABC: we can embed any single-winner instance into an ABC instance by simply adding $k-1$ dummy candidates and appropriately many voters (without uncertainty) who ensure that the dummy candidates are in every winning committee. Then, solving the possible committee problem in this ABC instance would also solve the possible-winner problem for the single-winner instance. Note that this extends to any ABC voting rule that is equivalent to AV for $k=1$.
Conversely, using \Cref{res:iterateApprCompletions} and techniques used by \citet{DBLP:conf/aaai/ImberIBK22} to find a necessary committee under partial orders, one can show that the necessary committee problem in our setting is tractable for all \textit{ABC scoring rules} (a large class of rules including all Thiele rules) for $d \leq 2$.

\section{Conclusions}\label{sec:conclusions}
We introduced the framework of partial spatial voting, where candidates and voters are positioned in a geometrical space, but voters can have intervals of possible values in each dimension/issue. For positional scoring rules, we recovered the tractable cases of necessary and possible winners in the model of partial orders, for every fixed number of issues. In particular, we showed that the possible winners can be found in polynomial time for the plurality and veto rules, and that the necessary winners can be found in polynomial time for every positional scoring rule. We identified cases where the possible-winner problem is hard for partial orders but not for partial spatial voting. Specifically, this holds for 
two-valued rules other than plurality and veto, such as $k$-approval and $k$-veto for $k>1$. We showed that the possible-winner problem may become intractable when the number of issues increases to a higher fixed number. 
For partial spatial approval ballots, we showed that the possible-winner problem under AV is intractable for all $d \geq 1$, and we gave an efficient algorithm for the necessary-winner problem under AV for $d \leq 2$.

It is left for future work to complete the picture of complexity for all positional scoring rules (see \Cref{tab:complexity}). Especially interesting is the case of Borda for $d=1$. It is also left to study the complexity of non-positional voting rules, as done in the case of the necessary and possible winners for model of partial orders~\cite{DBLP:journals/jair/XiaC11}.
It would also be interesting to study the implications of voter uncertainty on the strategic considerations of candidates. For instance, in a model where probabilistic information about voters' locations is available, candidates may want to relocate in order to increase their winning probability.

\section*{Acknowledgments}
The work of Markus Brill and Jonas Israel was supported by the Deutsche Forschungsgemeinschaft under grant BR 4744/2-1. The work of Aviram Imber and Benny Kimelfeld was supported by the Deutsche Forschungsgemeinschaft under project 412400621.

\bibliographystyle{plainnat}
\bibliography{References}

\end{document}